\documentclass[11pt,a4paper]{amsart}
\usepackage{amsmath,amssymb,amsthm, amscd,verbatim,enumerate}
\usepackage{graphicx,subfigure}
\usepackage{enumitem}
\usepackage[utf8]{inputenc}
\usepackage[lmargin=31mm,rmargin=31mm,bmargin=31mm,tmargin=31mm]{geometry}

\renewcommand{\baselinestretch}{1.1}
\allowdisplaybreaks

\newenvironment{enumeratei}{\begin{enumerate}[label=\textup{(\roman*)}, noitemsep, topsep=1.5mm, labelindent=.8em, leftmargin=*, widest=.]}{\end{enumerate}}
\renewenvironment{itemize}{\begin{enumerate}[label=$\bullet$, noitemsep, topsep=1.5mm, labelindent=.8em, leftmargin=*, widest=.]}{\end{enumerate}}

\newcommand{\norm}[1]{{|#1|}}

\theoremstyle{plain}
\newtheorem{theorem}{Theorem}[section]

\newtheorem{cor}[theorem]{Corollary}
\newtheorem{prop}[theorem]{Proposition}

\newtheorem{claim}[theorem]{Claim}

\theoremstyle{definition}

\def\C{1.677}

\title{Coloring Hypergraphs Induced by Dynamic Point Sets and Bottomless Rectangles}
\author{
Andrei Asinowski
\and 
Jean Cardinal
\and 
Nathann Cohen
\and 
S\'ebastien Collette
\and 
Thomas Hackl
\and 
Michael Hoffmann
\and 
Kolja Knauer
\and
Stefan Langerman
\and 
Micha{\l} Laso\'n
\and 
Piotr Micek
\and 
G\"unter Rote
\and 
Torsten Ueckerdt
}
\address{Freie Universit\"at, Berlin -- {\tt asinowski@mi.fu-berlin.de}}
\address{Universit\'e Libre de Bruxelles -- {\tt jcardin@ulb.ac.be}}
\address{Université Paris-Sud 11 -- {\tt nathann.cohen@gmail.com}}
\address{Universit\'e Libre de Bruxelles -- {\tt me@scollette.com}}
\address{TU Graz -- {\tt thackl@ist.tugraz.at}}
\address{ETH Z\"urich -- {\tt hoffmann@inf.ethz.ch}}
\address{TU Berlin -- {\tt knauer@math.tu-berlin.de}}
\address{Universit\'e Libre de Bruxelles -- {\tt slanger@ulb.ac.be}}
\address{Jagiellonian University in Krakow -- {\tt mlason@tcs.uj.edu.pl}}
\address{Jagiellonian University in Krakow -- {\tt piotr.micek@tcs.uj.edu.pl}}
\address{Freie Universit\"at, Berlin -- {\tt rote@inf.fu-berlin.de}}
\address{Charles University in Prague -- {\tt torsten@kam.mff.cuni.cz}}

\begin{document}
\pagestyle{empty}

\begin{abstract}
We consider a coloring problem on dynamic, one-dimensional point sets: points appearing and disappearing on a line at given times. We wish to color them with $k$ colors so that at any time, any sequence of $p(k)$ consecutive points, for some function $p$, contains at least one point of each color. 

We prove that no such function $p(k)$ exists in general. However, in the restricted case in which points appear gradually, but never disappear, we give a coloring algorithm guaranteeing the property at any time with $p(k)=3k-2$. 

This can be interpreted as coloring point sets in $\mathbb{R}^2$ with $k$ colors such that any bottomless rectangle containing at least $3k-2$ points contains at least one point of each color. Here a bottomless rectangle is an axis-aligned rectangle whose bottom edge is below the lowest point of the set.

For this problem, we also prove a lower bound $p(k)>ck$, where $c>1.67$. Hence for every $k$ there exists a point set, every $k$-coloring of which is such that there exists a bottomless rectangle containing $ck$ points and missing at least one of the $k$ colors.

Chen {\em et al.} (2009) proved that no such function $p(k)$ exists in
the case of general axis-aligned rectangles. Our result also
complements recent results from Keszegh and P\'alv\"olgyi on
cover-decomposability of octants (2011, 2012).  
\end{abstract}
\maketitle
\sloppy

\section{Introduction}

It is straightforward to color $n$ points lying on a line with $k$ colors in such a way that any set of $k$ consecutive points receive different colors; just color them cyclically with the colors $1,2,\dots,k, 1, \dots$. What can we do if points can appear and disappear on the line, and we wish a similar property to hold at any time? More precisely, we fix the number $k$ of colors, and wish to maintain the property that at any given time, any sequence of $p(k)$ consecutive points, for some function $p$, contains at least one point of each color.

We show that in general, such a function does not exist: there are dynamic point sets on a line that are impossible to color with two colors so that monochromatic subsequences have bounded length. This holds even if the whole schedule of appearances and disappearances is known in advance. This family of point sets is described in Section~\ref{sec:imposs}.

We prove, however, that there exists a linear function $p$ in the case where points can appear on the line at any time, but {\em never disappear}. Furthermore, this is achieved in a constructive, {\em semi-online} fashion: the coloring decision for a point can be delayed, but at any time the currently colored points yield a suitable coloring of the set. The algorithm is described in Section~\ref{sec:main}.

In Section~\ref{sec:bottom}, we restate the result in terms of a coloring problem in $\mathbb{R}^2$: for any integer $k\geq 1$, every point set in ${\mathbb R}^2$ can be colored with $k$ colors so that any {\em bottomless} rectangle containing at least $3k-2$ points contains one point of each color. Here, an axis-aligned rectangle is said to be bottomless whenever the $y$-coordinate of its bottom edge is $-\infty$. 

In Section~\ref{sec:lb}, we give lower bounds on the problem of coloring points with respect to bottomless rectangles. We show that the number of points $p(k)$ contained in a bottomless rectangle must be at least $1.67k$. 

Finally, in Section~\ref{sec:ck}, we consider an alternative problem in which we fix the size of the sequence to $k$, but we are allowed to increase the number of colors.

\subsection*{Motivations and previous works.}

The problem is motivated by previous intriguing results in the field of geometric hypergraph coloring. Here, a geometric hypergraph is a set system defined by a set of points and a set of geometric ranges, typically polygons, disks, or pseudodisks. Every hyperedge of the hypergraph is the intersection of the point set with a range.

It was shown recently~\cite{GV11} that for every convex polygon $P$, there exists a constant $c$, such that any point set in $\mathbb{R}^2$ 
can be colored with $k$ colors in such a way that any translation of $P$ containing at least $p(k)=ck$ points contains at least one point of each color.
This improves on several previous intermediate results~\cite{PT09,PT10,ACCLOR10}.

For the range spaces defined by translates of a given convex polygon, this corresponds to partitioning a given set of $n$ points into $k$ subsets, each subset being
an $\varepsilon$-net for $\varepsilon = ck/n$. More on the relation between this coloring problem and $\varepsilon$-nets can be found in the recent papers of 
Varadarajan~\cite{V10}, and Pach and Tardos~\cite{PT11}.

The problem for translates of polygons can be cast in its dual form as a covering decomposition problem: given a set of translates of a polygon $P$, we wish to
color them with $k$ colors so that any point covered by at least $p(k)$ of them is covered by at least one of each color. The two problems can be seen to be equivalent 
by replacing the points by translates of a symmetric image of $P$ centered on these points. The covering decomposition problem has a long history that dates back to conjectures by J\'anos Pach in the early 80s (see for instance \cite{P86,DG}, and references therein). The decomposability of coverings by unit disks was considered in a seemingly lost unpublished manuscript by Mani and Pach in 1986. Up to recently, however, surprisingly little was known about this problem. 

For other classes of ranges, such as axis-aligned rectangles, disks, translates of some concave polygons, or arbitrarily oriented strips~\cite{CPST09,PT10b,PTT05, P10}, such a coloring does not always exists, even when we restrict ourselves to two colors. For instance, the following result holds: for any integer $p\geq 2$, there exists a set of points in $\mathbb{R}^2$, every 2-coloring of which is such that there exists an open disk containing $p$ monochromatic points. Aloupis et al. gave positive results for axis-aligned strips~\cite{ACCIKLSST11}.

Keszegh~\cite{K07} showed in 2007 that every point set could be 2-colored so that any bottomless rectangle containing at least 4 points contains both colors. Later, Keszegh and P\'alv\"olgyi~\cite{KP11} proved the following cover-decomposability property of octants in $\mathbb R^3$: every collection of translates of the positive octant can be 2-colored so that any point of $\mathbb R^3$ that is covered by at least 12 octants is covered by at least one of each color. This result generalizes the previous one (with a looser constant), as incidence systems of bottomless rectangles in the plane can be produced by restricted systems of octants in $\mathbb R^3$. It also implies similar covering decomposition results for homothetic copies of a triangle. More recently, they generalized their result to $k$-colorings, and proved an upper bound of $p(k)<12^{2^k}$ on the corresponding function $p(k)$~\cite{KP12}. 

Our positive result on bottomless rectangles (Corollary~\ref{cor:main}) is a generalization of Keszegh's results~\cite{K07} to $k$-colorings. To our knowledge, this is the first example of a $k$-coloring achieving a linear bound on $p(k)$ for ranges that are not translates of a given convex body.

\section{Coloring dynamic point sets}
\label{sec:imposs}

A {\em dynamic point set} $S$ in $\mathbb{R}$ is a collection of triples $(v_i, a_i, d_i)\in\mathbb{R}^3$, with $d_i\geq a_i$, that is interpreted as follows: the point $v_i\in\mathbb{R}$ appears on the real line at time $a_i$ and disappears at time $d_i$. Hence, the set $S(t)$ of points that are present at time $t$ are the points $v_i$ with $t\in [a_i,d_i)$. A $k$-coloring of a dynamic point set assigns one of $k$ colors to each such triple.

We now show that it is not possible to find a 2-coloring of such a point set while avoiding long monochromatic subsequences at any time.

\begin{theorem}
\label{thm:imposs}
For every $p\in\mathbb{N}$, there exists a dynamic point set $S$ with the following property: for every 2-coloring of $S$, there exists a time $t$ such that $S(t)$ contains $p$ consecutive points of the same color.  
\end{theorem}

\begin{proof}
In order to prove this result, we work on an equivalent two-dimensional version of the problem. From a dynamic point set, we can build $n$ horizontal segments in the plane, where the $i$th segment goes from $(a_i,v_i)$ to $(d_i,v_i)$. At any time $t$ the visible points $S(t)$ correspond to the intervals that intersect the line $x=t$. It is therefore equivalent, in order to obtain our result, to build a collection of horizontal segments in the plane that cannot be 2-colored in such a way that any set of $p$ segments intersecting some vertical segment contains one element of each color.

Our construction borrows a technique from Pach, Tardos, and T\'oth~\cite{PTT05}. In this paper, the authors provide an example of a set system whose base set cannot be 2-colored without leaving some set monochromatic. This set system $\mathcal S$ is built on top of the $1+p+\dots+p^{p-1}=\frac {1-p^p}{1-p}$ vertices of a $p$-regular tree $\mathcal T^p$ of depth $p$, and contains two kinds of sets~:
\begin{itemize}
\item the $1+p+\dots+p^{p-2}$ sets of {\it siblings}: the sets of $p$ vertices having the same father,
\item the $p^{p-1}$ sets of $p$ vertices corresponding to a path from the root vertex to one of the leaves in $\mathcal T^p$.
\end{itemize}
It is not difficult to realize that this set system is not 2-colorable: by contradiction, if every set of siblings is non-monochromatic, we can greedily construct a monochromatic path from the root to  a leaf.

We now build a collection of horizontal segments corresponding to the vertices of $\mathcal T^p$, in such a way that for any set $E\in \mathcal S$ there exists a time $t$ at which the elements of $E$ are consecutive among those that intersect the line $x=t$. For any $p$ (see \figurename~\ref{fig:imposs}), the construction starts with a building block $B^1_p$ of $p$ horizontal segments, the $i$th segment going from $(-\frac ip,i)$ to $(0,i)$. Because these $p$ segments represent {\it siblings} in $\mathcal T^p$, they are consecutive on the vertical line that goes through their rightmost endpoint, and hence cannot all receive the same color.



Block $B^{j+1}_p$ is built from a copy of $B^1_p$ to which are added
$p$ resized and translated copies of $B^j_p$ : the $i$th copy lies in
the rectangle with top-right corner $(-\frac {i-1} p,i+1)$ and
bottom-left corner $(-\frac i p,i)$.
By adding to $B_p^{p-1}$ a last horizontal segment below all others,
corresponding to the root of $\mathcal T^p$, the ancestors of a
segment are precisely those that are below it on the vertical line
that goes through its leftmost point.
When such sets of ancestors are of cardinality $p-1$, which only
happens when one considers the set of ancestors of a leaf, then the
set formed by the leaf and its ancestors is required to be
non-monochromatic.

With this construction we ensure that a feasible 2-coloring of the
segments would yield a proper 2-coloring of $\mathcal S$, which we
know does not exist.
\end{proof}

\begin{figure}[h!]
  \vspace{.2cm}
  \centering
  \subfigure[The tree ${\mathcal T}^3$.]{\includegraphics[scale=.5]{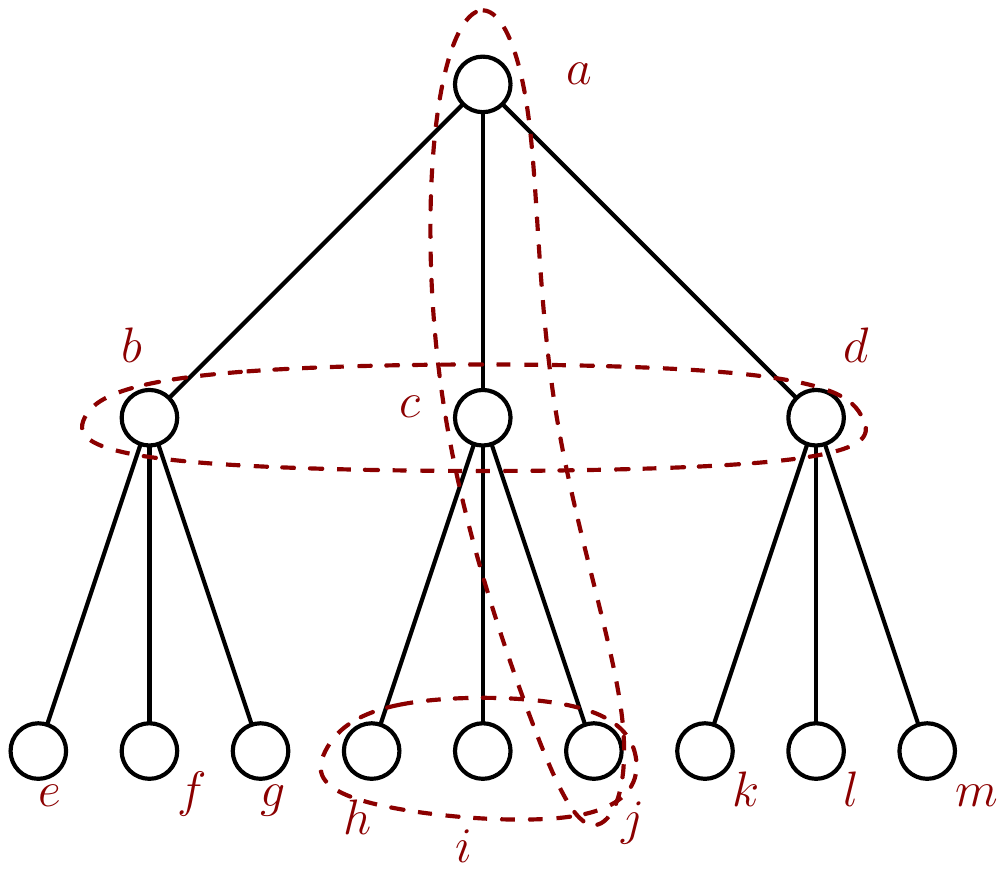}}
\hspace{.2cm}
  \subfigure[The corresponding set of horizontal segments $B_3^2$, with a root segment $a$.]{\includegraphics[scale=.5]{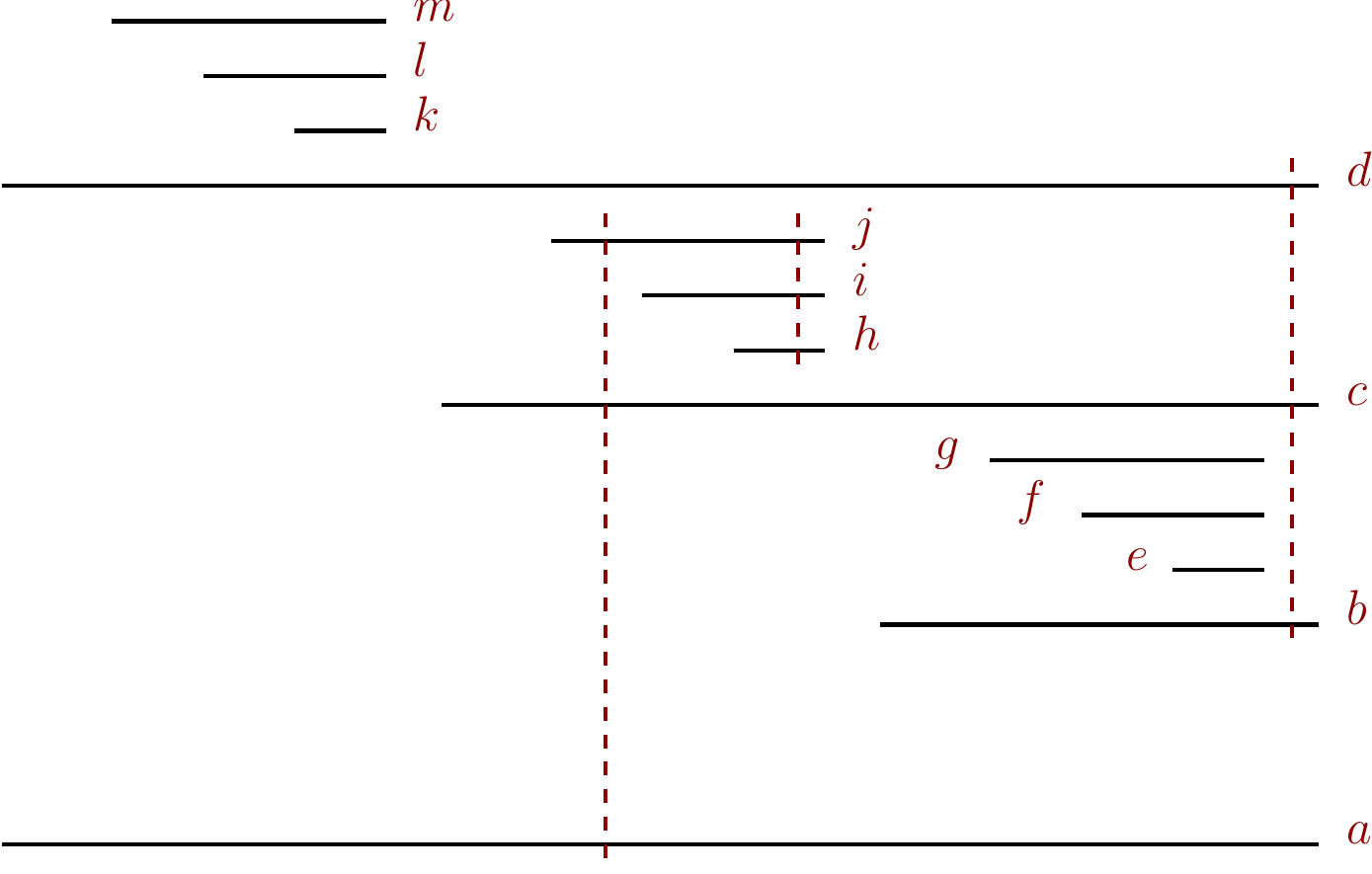}}
  \caption{The recursive construction of theorem~\ref{thm:imposs}, for $p=3$.}
  \label{fig:imposs}
\end{figure}

The above result implies that no function $p(k)$ exists for any $k$ that answers the original question. If it were the case, then we could simply merge
color classes of a $k$-coloring into two groups and contradict the above statement.

\begin{figure}
\begin{center}
\includegraphics[width=.4\textwidth]{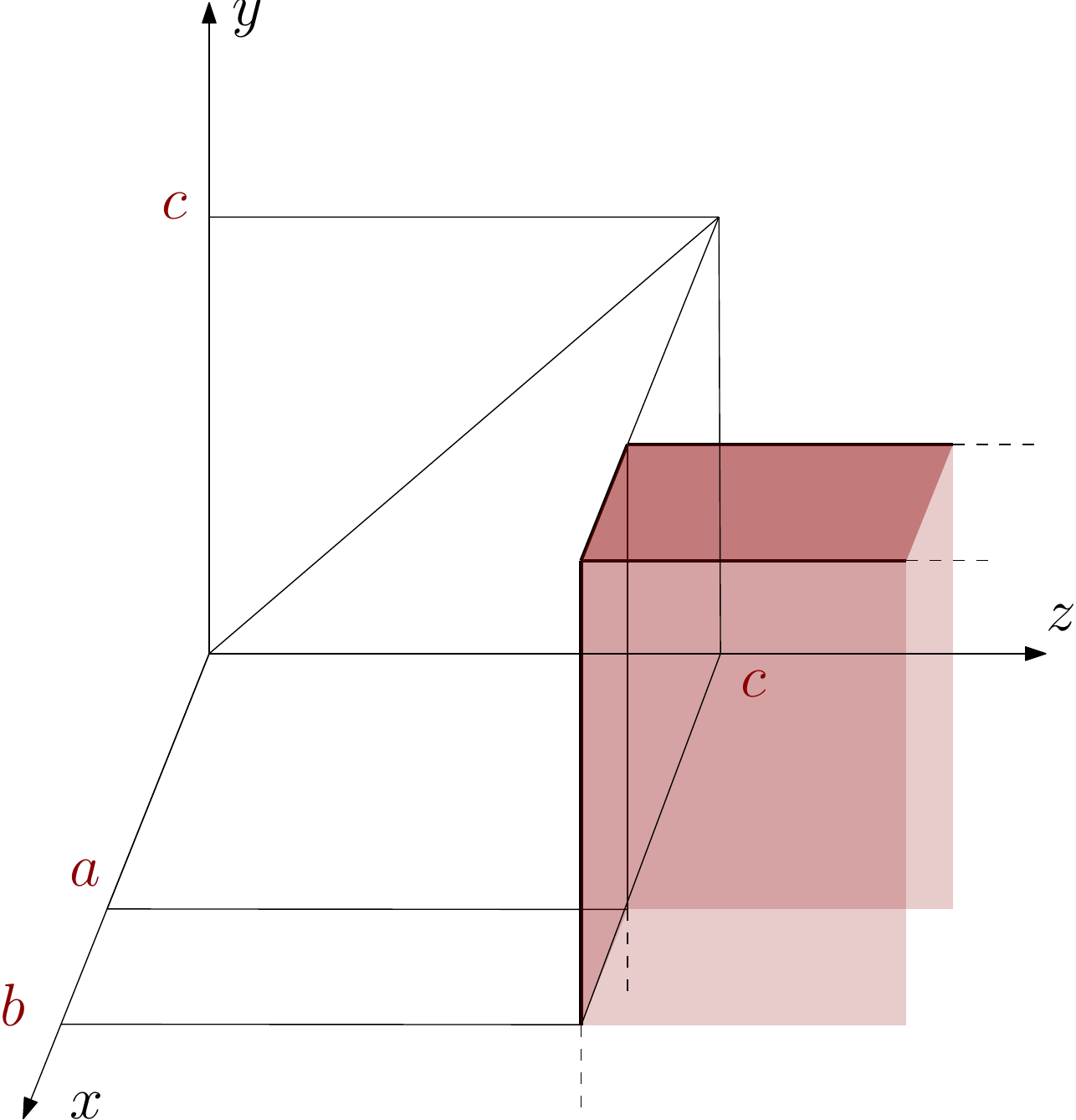}
\end{center}
\caption{\label{fig:corner}A corner with coordinates $(a,b,c)$.}
\end{figure}

Theorem~\ref{thm:imposs} can also be interpreted as the indecomposability of coverings by a specific class of unbounded polytopes in $\mathbb{R}^3$. We define a {\em corner} with coordinates $(a,b,c)$ as the following subset of $\mathbb{R}^3$: $\{(x,y,z)\in\mathbb{R}^3 : a\leq x\leq b, y\leq c, z\geq c \}$. An example is given in \figurename~\ref{fig:corner}. One can verify that a point $(x,y,z)$ is contained in a corner
$a,b,c$ if and only if the vertical line segment with endpoints $(x,y)$ and $(x,z)$ intersects the horizontal line segment with endpoints $(a,c)$ and $(b,c)$. The corollary follows.
\begin{cor}
\label{cor:indec}
For every $p\in\mathbb{N}$, there exists a collection $S$ of corners with the following property: for every 2-coloring of $S$, there exists a point $x\in\mathbb{R}^3$ contained in exactly $p$ corners of $S$, all of the same color. In other words, corners are not cover-decomposable. 
\end{cor}

\section{Coloring point sets under insertion}
\label{sec:main}

Since we cannot bound the function $p(k)$ in the general case, we now consider a simple restriction on our dynamic point sets: we let the deletion times $d_i$ be infinite for every $i$. Hence, points appear on the line, but never disappear.

A natural idea to tackle this problem is to consider an online coloring strategy, that would assign a color to each point in order of their arrival times $a_i$, without any knowledge of the points appearing later. However, we cannot guarantee any bound on $p(k)$ unless we delay some of the coloring decisions. To see this, consider the case $k=2$, and call the two colors red and blue. An online algorithm must color each new point in red or blue as soon as it is presented. We can design an adversary such that the following invariant holds: at any time, the set of points is composed of a sequence of consecutive red points, followed by a sequence of consecutive blue points. The adversary simply chooses the new point to lie exactly between the two sequences at each step.

Our computation model will be {\em semi-online}: The algorithm considers the points in their order of the arrival time $a_i$. At any time, a point in the sequence either has one of the $k$ colors, or is uncolored. Uncolored points can be colored later, but once a point is colored, it keeps its color for the rest of the procedure. At any time, the colors that are already assigned suffice to satisfy the property that any subsequence of $3k-2$ points has one point of each color, i.e., $p(k)\leq 3k-2$. 

\begin{theorem}
\label{thm:main}
Every dynamic point set without disappearing points can be $k$-colored in the semi-online model such that at any time, every subsequence of at least $3k-2$ consecutive points contains at least one point of each color.
\end{theorem}

\begin{proof}
We define a {\em gap for color} $i$ as a maximal interval (set of consecutive points) containing no point of color $i$, 
that is, either between two successive occurrences of color~$i$, 
or before the first occurrence (first gap), 
or after the last occurrence (last gap),
or the whole line if no point has color $i$. 
A {\em gap} is simply a gap for color $i$, for some $1\leq i\leq k$. 
We propose an algorithm for a semi-online model keeping the sizes of all gaps to be at most $3k-3$. 
This means every set of $3k-2$ consecutive points contains each color at least once and implies $p(k)\leq 3k-2$.
The algorithm maintains two invariants:\\

{\it (a) every gap contains at most $3k-3$ points;}\\

{\it (b) if there is some point colored with $i$ then every gap for
  color $i$, except the first and the last gap, contains at least $k-1$ points.}\\


The two invariants are vacuous when the set of points is empty. 
Now, suppose that the invariants hold for an intermediate set of points and consider a new point on the line presented by an adversary. 
Clearly, invariant {\it (b)} cannot be violated in the extended set as no gaps decrease in size. 
However, there may arise some gaps of size $3k-2$ violating {\it (a)}. 
If not then the invariants hold for the extended set and the algorithm does not color any point in this step.
Suppose there are some gaps of size $3k-2$, consider one of them, say a gap of color $i$, and denote the points in the gap in their natural ordering on the line from left to right as $(\ell_1, \ldots, \ell_{k-1}, m_1, \ldots, m_k, r_1,\ldots r_{k-1})$. 
Now, color $i$ does not appear among these points. 
Invariant {\it (b)} yields that none of the $k-1$ remaining colors appears twice among $m_1, \ldots, m_k$. 
Thus, there is some $m_j$, which is uncolored and the algorithm colors it with $i$. 
This splits the large gap into two smaller gaps. 
Moreover, since there are $k-1$ $\ell$-points and $k-1$ $r$-points invariant {\it (b)} is maintained for both new $i$-gaps. 
The algorithm repeats that process until all gaps are of size at most $3k-3$.

This concludes the proof, as after the algorithm ends all remaining
uncolored points can be arbitrary colored.
\end{proof}

\section{Coloring points with respect to bottomless rectangles}
\label{sec:bottom}

A {\em bottomless rectangle} is a set of the form $\{ (x,y)\in \mathbb{R}^2 : a \leq x \leq b, y \leq c \}$, for a triple of real numbers $(a,b,c)$ with $a\leq b$. We consider the following geometric coloring problem: given a set of points in the plane, we wish to color them with $k$ colors so that any bottomless rectangle containing at least $p(k)$ points contains at least one point of each color. It is not difficult to realize that the problem is equivalent to that of the previous section.

\begin{cor}
\label{cor:main}
Every point set $S\subset \mathbb{R}^2$ can be colored with $k$ colors so that any bottomless rectangle containing at least $3k-2$ points of $S$ contains at least one point of each color.
\end{cor}
\begin{proof}
The algorithm proceeds by sweeping $S$ vertically in increasing $y$-coordinate order. This defines a dynamic point set $S'$ that contains at time $t$ the $x$-coordinates of the points below the horizontal line of equation $y=t$. The set of points of $S$ that are contained in a bottomless rectangle $\{ (x,y)\in \mathbb{R}^2 : a \leq x \leq b, y \leq t \}$ correspond to the points in the interval $[a,b]$ in $S'(t)$. Hence, the two coloring problems are equivalent, and Theorem~\ref{thm:main} applies.
\end{proof}

\section{Lower Bound}
\label{sec:lb}

We now give a lower bound on the smallest possible value of $p(k)$.

\begin{theorem}
For any $k$ sufficiently large, there exists a point set $P$ such that for any $k$-coloring of $P$, 
there exists a color $i\in [k]$ and a bottomless rectangle containing at least $\C k - 2.5$ points, 
none of which are colored with color $i$.
\end{theorem}
\begin{proof}
Fix $k \geq 100$.
For $n \in \mathbb{N}$ and $0\leq a < k$ we define the point set $P = P(n,a)$ to be the union of point sets $L$, $R$ and $B$ (standing for left, right and bottom, respectively) as follows:
\begin{align*}
 L &:= \{ (i-n,2i) \in \mathbb{R}^2 \;|\; i \in [n] \} \\
 B &:= \{ (i,0) \in \mathbb{R}^2 \;|\; i \in [a] \} \\
 R &:= \{ (n+a+i,2n+1-2i) \in \mathbb{R}^2 \;|\; i \in [n] \}
\end{align*}
See Figure~\ref{fig:LB_pointset} for an illustration.
Note that $|L| = |R| = n$ and $|B| = a$.

\begin{figure}[htb]
 \centering
 \subfigure[\label{fig:LB_pointset}]
 {\includegraphics{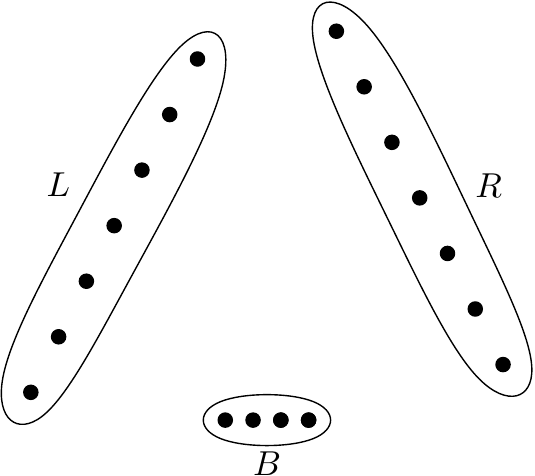}}
 \hfill
 \subfigure[\label{fig:LB_rectangles}]
 {\includegraphics{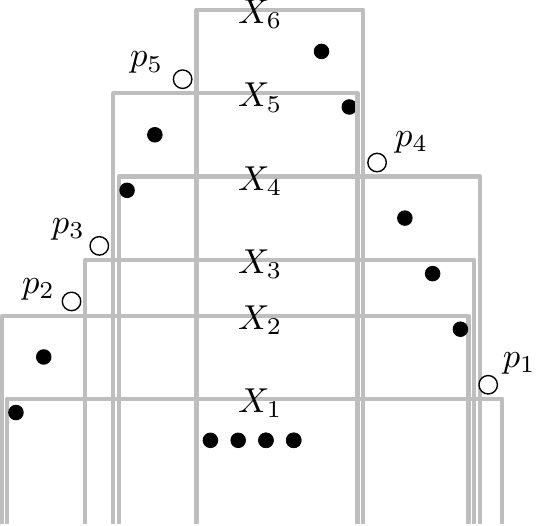}}
 \caption{(a) The point set $P = P(n,a)$ with $n = 7$ and $a = 4$, and (b) the bottomless rectangles $X_1,\ldots,X_6$ corresponding to the color class $P(c^*) = \{p_1,\ldots,p_5\}$.}
\end{figure}

Consider any coloring of the points in $P$ with colors from $[k]$. 
For a color $i\in[k]$ we define $P(i)$ to be the subset of points of $P$ colored with $i$. 
We assume for the sake of contradiction that every bottomless rectangle that contains $b := \lfloor \C k - 2.5\rfloor$ points, contains one point of each color.
In the remainder of the proof we will identify a bottomless rectangle containing $b'$ points but no point of one particular color.
We give a lower bound for $b'$ depending on $n$ and $a$, but independent of the fixed coloring under consideration. 
Taking sufficiently large $n$ and choosing $a = \lfloor 0.655 k \rfloor$ we will prove $b' > b$, which contradicts our assumption and hence concludes the proof.

A color used at least once for the points in $B$ is called a \emph{low color} and a point colored with a low color is a \emph{low point}. 
Note that there are low points outside of the set $B$.
Let $\ell$ be the number of low colors. 
Clearly, $\ell \leq \norm{B} = a$.

\begin{claim}\label{claim:green-colors}
\hfill
\begin{enumeratei}
  \item For every non-low color $c$ there are at least $\left\lfloor\frac{n}{b-a}\right\rfloor$ points of color $c$ in $L$.\label{enum:non-green-in-L}
  \item There are at least $\sum\limits_{i=0}^{\ell-1}\left\lfloor\frac{n}{b-i}\right\rfloor$ low points in $L$.\label{enum:green-in-L}
 \end{enumeratei}
\end{claim}
\begin{proof}[Proof of Claim~\ref{claim:green-colors}]
Fix a color $c\in[k]$ and assume that the $j$ leftmost points in $B$ are not colored with $c$. 
Order the points in $L$ colored with $c$ according to their $x$-coordinate: $p_1$, $p_2$,\ldots, $p_m$. 
Now for each $1<i\leq m$ there is a bottomless rectangle containing all points in $L$ between $p_{i-1}$ and $p_i$, and the leftmost $j$ points in $B$, and nothing else. 
Additionally, there is a bottomless rectangle containing all points in $L$ to the left of $p_1$ together with $j$ leftmost points in $B$, and a bottomless rectangle containing all points in $L$ to the right of $p_m$ together with $j$ leftmost points in $B$. 
Note that all these rectangles are disjoint within $L$ and each point from $L$ not colored with $c$ lies in exactly one such rectangle.
Since each such rectangle $X$ avoids the color $c$ we get that $\norm{X\cap P}\leq b - 1$ and $\norm{X\cap L}\leq b-1-j$ and therefore
\begin{gather}
 m+ (m+1)(b-1-j) = m(b - j) + b - j -1 \geq \norm{L} = n,\notag\\
 m\geq \left\lfloor \frac{n}{b - j}\right\rfloor.\label{eq:hmmm}
\end{gather}

In order to prove \ref{enum:non-green-in-L} consider a non-low color $c$. 
As $c$ is not used on points in $B$ at all we can put $j=a$ in \eqref{eq:hmmm} and the statement of \ref{enum:non-green-in-L} follows. 
Now, if $c$ is a low color, then $j$ defined as the maximum number of leftmost points in $B$ avoiding $c$ is always less than $a$.
However for different low colors $c$ we obtain different $j$.
Thus the sum of inequality~\eqref{eq:hmmm} over all low colors is minimized by $\sum_{i=0}^{\ell-1}\lfloor \frac{n}{b-i}\rfloor$, which gives~\ref{enum:green-in-L}.
\end{proof}

By Claim~\ref{claim:green-colors} \ref{enum:non-green-in-L} and \ref{enum:green-in-L} combined we get that there is a set $S$ of $k-a$ non-low colors such that at most $n - \sum_{i=0}^{a-1}\lfloor\frac{n}{b - i}\rfloor$ points in $L$ have a color from $S$.
Analogously, at most $n - \sum_{i=0}^{a-1}\lfloor\frac{n}{b - i}\rfloor$ in $R$ have a color from $S$.
Summing up we get:
\begin{align*}
\sum_{c \in S}|P(c)| &= \sum_{c \in S}\bigl(|P(c) \cap L| + |P(c) \cap R|\bigr) \\
&\leq 2n - 2\sum_{i=0}^{a-1}\left\lfloor\frac{n}{b - i}\right\rfloor\\
&\leq 2n - 2\sum_{i=0}^{a-1}\left(\frac{n}{b - i}-1\right)\\
& = 2n\left(1 - \sum_{i=b-a+1}^{b}\frac{1}{i}\right)+2a\\
& = 2n\left(1 - \sum_{i=1}^{b}\frac{1}{i} + \sum_{i=1}^{b-a}\frac{1}{i}\right)+2a.
\end{align*}
Using that $\sum_{i=1}^{x}\frac{1}{i} = \ln(x+1) - \sum_{j=1}^{\infty} \frac{B_j}{j(x+1)^j} + \gamma$ for every $x \geq 1$, 
where $B_j$ are the Bernoulli numbers of second kind and $\gamma$ is the Euler-Mascheroni constant, we obtain
\begin{align*}
\sum_{c \in S}|P(c)| &< 2n\left(1 - \ln(b+1) + \ln(b-a + 1)\right)+2a \\
& = 2n\left(1 - \ln\left(\frac{b+1}{b-a+1}\right)\right) + 2a.
\end{align*}
From the pigeonhole principle there is a color $c^* \in S$, such that
\begin{equation}
 q := |P(c^*)| \leq \left\lfloor \frac{2n(1 - \ln(\frac{b+1}{b-a+1})) + 2a}{k-a} \right\rfloor.\label{eq:color-class}
\end{equation}
Enumerate the points in $P(c^*)$ by $p_1,p_2,\ldots,p_q$ according to their increasing $y$-coordinates, i.e., we have $i < j$ if{f} $p_i$ has smaller $y$-coordinate than $p_j$.
Now we consider all maximal bottomless rectangles that completely contain $B$ and contain no point of color $c^*$.
There are exactly $q+1$ such rectangles:
For every point $p_i \in P(c^*)$ there is a bottomless rectangle $X_i$ whose top side lies immediately below $p_i$.
And one further bottomless rectangle $X_{q+1}$ containing the entire strip between $L$ and $R$, and with sides bounded by the point in $P(c^*) \cap L$ and the point in $P(c^*) \cap R$ with the highest index. See Figure~\ref{fig:LB_rectangles} for an illustration.

\begin{claim}\label{claim:3-halves}
 \begin{equation*}
  \sum_{i=1}^{q}|X_i\cap (L\cup R)| \geq \frac{3}{2}\bigl(2n - q - b + a\bigr).\label{eq:3-halves}
 \end{equation*}
\end{claim}
\begin{proof}[Proof of Claim~\ref{claim:3-halves}]
 We partition the points in $L \cup R$ that are not colored $c^*$ into $q+1$ sets $Y_1,\ldots,Y_{q+1}$ with consecutive $y$-coordinates, 
 such that $Y_i \subset X_i \cap (L\cup R)$ for all $i=1,\ldots,q+1$.
 Clearly, $|X_i \cap Y_i| = |Y_i|$.
 We claim that $|X_{i+1} \cap Y_i| \geq \frac{1}{2}|Y_i|$, for $i = 1,\ldots,q$.
 Without loss of generality, let us assume that $p_i \in L$.
 Then either $Y_i=\emptyset$ or the point in $Y_i$ with largest $y$-coordinate lies in $R$.
 Since points from $L$ and $R$ alternate in the ordering of $L\cup R$ with respect to increasing $y$-coordinate it follows that $Y_i$ is almost equally 
partitioned into its left part $Y_i \cap L$ and its right part $Y_i \cap R$.
 Since the topmost point in $Y_i$ lies in $R$ we have $|Y_i \cap R| \geq \frac{1}{2}|Y_i|$.
 Now since $p_i\in L$ we have $X_{i+1} \supset Y_i\cap R$, and thus
 \begin{equation}
  |X_{i+1} \cap Y_i| \geq |Y_i \cap R| \geq \frac{1}{2}|Y_i|. \label{eq:1halfYi}
 \end{equation}
Note also that $|X_{q+1}\cap Y_q|< b-a$ as $X_{q+1}$ avoids color $c^*$, so $|X_{q+1}|<b$, and contains all $a$ points in $B$.
Now we calculate
\begin{align*}
 \sum_{i=1}^{q}|X_i\cap (L\cup R)| &\geq \bigl(\sum_{i=1}^{q}|X_i \cap Y_i| + |X_{i+1} \cap Y_i|\bigr) - |X_{q+1} \cap Y_q|\\
 &\overset{\eqref{eq:1halfYi}}{\geq} \sum_{i=1}^{q} \frac{3}{2}|Y_i| - |X_{q+1} \cap Y_{q}|\\
 &= \frac{3}{2}\bigl(2n-|P(c^*)|-|Y_{q+1}|\bigr) - |X_{q+1} \cap Y_{q}|\\
 &\geq \frac{3}{2}\bigl(2n-q-|X_{q+1} \cap (L \cup R)|\bigr) \geq \frac{3}{2}\bigl(2n-q-(b-a)\bigr).
\end{align*}
\end{proof}

From Claim~\ref{claim:3-halves} we get from the pigeonhole principle that there is a bottomless rectangle $X^* \in \{X_1,\ldots,X_q\}$ with
\begin{align*}
 |X^*| &\geq \frac{\frac{3}{2}(2n - q - b + a)}{q} + a\\
 &= \frac{3n}{q} - \frac{3}{2} - \frac{3(b-a)}{2q} + a\\
 &\overset{\eqref{eq:color-class}}{\geq} \frac{3(k-a)}{2\bigl(1-\ln\bigl(\frac{b+1}{b-a+1}\bigr) + \frac{2a}{n}\bigr)} + a - \frac{3}{2} - \frac{3(b-a)}{2q}
\end{align*}
Now, if we increase $n$, then $q = |P(c^*)|$ increases as well, and for sufficiently large $n$ the terms $\frac{2a}{n}$ in the denominator and the additive term $\frac{3(b-a)}{2q}$ become negligible.
In particular, with $a:= \lfloor 0.655k\rfloor$ and $b = \lfloor \C k -2.5\rfloor$ and sufficiently large $n$ we have
\begin{align*}
 |X^*| &\geq \frac{3(k-a)}{2\bigl(1-\ln\bigl(\frac{b+1}{b-a+1}\bigr)\bigr)} + a - \frac{3}{2}\\
 &= \frac{3(k-\lfloor 0.655k\rfloor)}{2\bigl(1-\ln\bigl(\frac{\lfloor \C k - 2.5\rfloor+1}{\lfloor \C k - 2.5\rfloor-\lfloor 0.655k\rfloor+1}\bigr)\bigr)} + \lfloor 0.655k\rfloor - \frac{3}{2} \\
 &\geq \frac{3(k-0.655k-1)}{2\bigl(1-\ln\bigl(\frac{\C k - 2.5-1+1}{\C k - 2.5-0.655k-1+1}\bigr)\bigr)} + 0.655k -1 - \frac{3}{2}\\
 &= \frac{1.035k - 3}{2\bigl(1-\ln\bigl(\frac{\C k - 2.5}{1.022k - 2.5}\bigr)\bigr)} + 0.655k - 2.5.
\end{align*}
For $k \to \infty$ (starting with $k \geq 3$) the above expression is monotonously increasing and its limit is given by
\[
 \Bigl(\frac{1.035}{2\bigl(1-\ln\bigl(\frac{\C}{1.022}\bigr)\bigr)} + 0.655 \Bigr)k > 1.68k.
\]
Hence if $k$ is big enough ($k \geq 100$ is actually enough) the bottomless rectangle $X^*$ contains strictly more than $\C k -2.5$ points but no point of color $c^*$, which is a contradiction and concludes the proof.
\end{proof}

\section{Increasing the Number of Colors}
\label{sec:ck}

There is another problem which can be tackled this time in an {\em online} model. The number $c(k)$ is the minimum number of colors needed to color the points on a line such that any set of at most $k$ consecutive points is completely colored by distinct colors. The same problem has been considered for other types of geometric hypergraphs by Aloupis et al.~\cite{ACCLS09}. 

Again, the algorithm considers the points in their order of the arrival time $a_i$ but now colors them immediately. At any time, a point in the sequence has one of $2k-1$ colors. The colors that are already arrived satisfy the property that any subsequence of $k$ points has no color twice, i.e.,  $c(k)\leq 2k-1$. 

\begin{prop}
\label{prop:c(k)}
Every dynamic point set without disappearing points can be $(2k-1)$-colored in the online model such that at any time, every subsequence of at least $k$ consecutive points contains no color twice.
\end{prop}
\begin{proof}
At the arrival of a new point $p$ denote by $(\ell_1,\ldots, \ell_{k-1})$ and $(r_1,\ldots, r_{k-1})$ the $k-1$ points to its left and to its right, respectively. Together they have at most
$2k-2$ colors, Thus, there is at least one of the $2k-1$ colors unused among these points. The algorithm colors $p$ with this color.
\end{proof}

Now, we consider the following geometric coloring problem: given a set of points in the plane, we wish to color them with $c(k)$ colors so that any bottomless rectangle containing at most $k$ points contains at most one point of each color and all points are colored. It is not difficult to realize that the problem is equivalent to that of the previous section. One can use Proposition~\ref{prop:c(k)} to show:
\begin{cor}
\label{cor:c(k)}
Every point set $S\subset \mathbb{R}^2$ can be colored with $2k-1$ colors so that any bottomless rectangle containing at least $k$ points of $S$ contains no color twice.
\end{cor}

The number of colors used in Corollary~\ref{cor:c(k)} is smallest possible. This is witnessed by a point set $S$ consisting of $k$ points of the form $\{(i,2i)\mid 0\leq i\leq k-1\}$
and $k-1$ points of the form $\{(2k-i,2i-1)\mid 1\leq i\leq k-1\}$, see \figurename~\ref{fig:ck_lb} for an example. It is easy to see that every pair of points in such a point set is in a common bottomless rectangle of size at most $k$.

\begin{figure}[h!]
  \centering
  {\includegraphics[scale=.5]{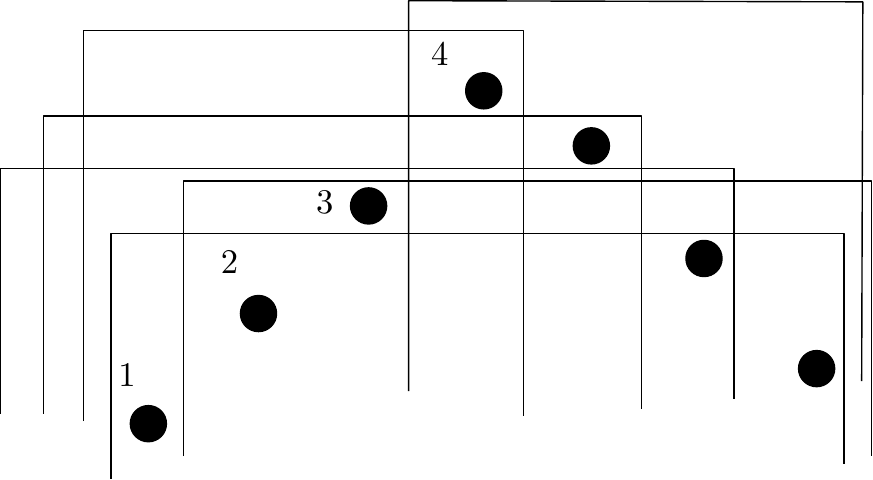}}
  \caption{A point set witnessing $c(k)\geq 2k-1$ for $k=4$.}
  \label{fig:ck_lb}
\end{figure}

Finally, let us remark that an upper bound on $c(k)$ for dynamic point sets in which points can both appear and disappear, as in Section~\ref{sec:imposs}, can be obtained by bounding 
the chromatic number of the corresponding so-called {\em bar $k$-visibility graph}, as defined by Dean et al.~\cite{DEGLST07}. In particular, they show that those graphs have $O(kn)$ edges, yielding $c(k)=O(k)$
for that case.

\section*{Acknowledgments}

This research is supported by the the ESF EUROCORES programme EuroGIGA, CRP ComPoSe ({\tt http://www.eurogiga-compose.eu}).
It was initiated at the ComPoSe kickoff meeting held at CIEM (International Centre for Mathematical meetings) in Castro de Urdiales (Spain) on May 23--27, 2011, and pursued
at the 2nd ComPoSe Workshop on Geometric Graphs and Order Types held at TU Graz (Austria) on April 16--20, 2012. The authors warmly thank the organizers of these two meetings: 
Oswin Aichholzer, Ferran Hurtado, Paco Santos, and Birgit Vogtenhuber, as well as all the other participants, for providing such a great working environment. A preliminary version
of these results were presented by a subset of the authors to the European Workshop on Computational Geometry, held in Assisi (Italy) on March 19--21, 2012.

\small
\bibliographystyle{plain}
\bibliography{coloring}

\end{document}